\documentclass[journal]{IEEEtran}
\ifCLASSINFOpdf
  \usepackage[pdftex]{graphicx}
   \graphicspath{{../figures/}{../jpeg/}}
\else
\fi

\usepackage[cmex10]{amsmath}
\usepackage{amssymb}
\usepackage{bm}

\usepackage{algpseudocode}
\usepackage{algorithm}
\usepackage{array}
\usepackage{url}
\usepackage[noadjust]{cite}
\newtheorem{theorem}{Theorem}

\newtheorem{lemma}{Lemma}

\newtheorem{proposition}{Proposition}
\newtheorem{remark}{Remark}

\usepackage{fancyhdr}
\usepackage{amsfonts}
\usepackage{xfrac}
\usepackage{color}
\usepackage{pgfplots}
\pgfplotsset{compat=newest}           
\pgfplotsset{compat=1.17}
\usepackage{enumerate}

\usepackage{cite}
\usepackage{mathrsfs}
\usepackage{subcaption}
\usepackage{stfloats}
\usepackage{makecell}
\usepackage{comment}
\usepackage{afterpage}
\usepackage{tikz}
\usepackage{mathdots}
\usepackage{yhmath}
\usepackage{cancel}
\usepackage{color}
\usepackage{siunitx}
\usepackage{array}
\usepackage{multirow}
\usepackage{textcomp}
\usepackage{gensymb}
\usepackage{tabularx}
\usepackage{booktabs}
\usepackage{graphicx}
\usepackage{xcolor}

\allowdisplaybreaks
\usetikzlibrary{decorations.markings,decorations.pathmorphing}

\begin{document}
    \title{Location-Driven Programmable Wireless Environments through Light-emitting RIS (LeRIS)}
  \author{Dimitrios~Bozanis,~\IEEEmembership{Student Member,~IEEE,} Dimitrios~Tyrovolas,~\IEEEmembership{Member,~IEEE,} \\
      Vasilis K. Papanikolaou,~\IEEEmembership{Member,~IEEE,} 
      Sotiris A. Tegos,~\IEEEmembership{Senior Member,~IEEE,} \\ Panagiotis D. Diamantoulakis,~\IEEEmembership{Senior Member,~IEEE,} Christos K. Liaskos,~\IEEEmembership{Member,~IEEE,} \\ Robert Schober,~\IEEEmembership{Fellow,~IEEE},
      and~George K.~Karagiannidis,~\IEEEmembership{Fellow,~IEEE}

\thanks{D. Bozanis, D. Tyrovolas, S. A. Tegos, P. D. Diamantoulakis, and G. K. Karagiannidis are with the Aristotle University of Thessaloniki, 54124 Thessaloniki, Greece (dimimpoz@auth.gr, tyrovolas@auth.gr, tegosoti@auth.gr, padiaman@auth.gr, geokarag@auth.gr).}
\thanks{V. K. Papanikolaou and R. Schober are with the Friedrich Alexander University, 91058 Erlangen, Germany (vasilis.papanikolaou@fau.de, robert.schober@fau.de).}
\thanks{C. K. Liaskos is with the University of Ioannina, Greece, and with the Foundation for Research and Technology Hellas (FORTH), Greece (cliaskos@uoi.gr).}
}


\maketitle

\begin{abstract}
As 6G wireless networks seek to enable robust and dynamic programmable wireless environments (PWEs), reconfigurable intelligent surfaces (RISs) have emerged as a cornerstone for controlling electromagnetic wave propagation. However, realizing the potential of RISs for demanding PWE applications depends on precise and real-time user localization, especially in scenarios with random receiver orientations and inherent hardware imperfections. To address this challenge, we propose a novel optical localization framework that integrates conventional ceiling-mounted LEDs with light-emitting reconfigurable intelligent surfaces (LeRISs). By leveraging the spatial diversity offered by the LeRIS architecture, the framework introduces robust signal paths that improve localization accuracy and reduce errors under varying orientations. To this end, we derive a system of equations for received signal strength-based localization that accounts for random receiver orientations and imposes spatial constraints on LED placement, ensuring unique and reliable solutions. Finally, our simulation results demonstrate that the proposed framework achieves precise beam control and high spectral efficiency even for RISs with large number of reflecting elements, establishing our solution as scalable and adaptive for PWEs that require real-time accuracy and flexibility.

\end{abstract}


%
\IEEEpeerreviewmaketitle

\section{Introduction}
As wireless networks evolve toward the sixth generation (6G), which is expected to support applications such as immersive metaverse, smart homes, and advanced healthcare facilities, the focus has shifted from optimizing individual communication links to transforming the wireless propagation environment into a controllable and intelligent entity \cite{metaverse}. This vision is realized through programmable wireless environments (PWEs), where electromagnetic wave (EM) propagation is dynamically adapted to user needs, providing personalized, high-quality services \cite{PWE1}. At the core of this vision are reconfigurable intelligent surfaces (RISs), which are designed to manipulate the reflection of incident EM waves \cite{PWE1,RISsurvey}. Specifically, RISs can be strategically placed throughout the environment and interact with incoming waves through multiple controllable reflecting elements, dynamically adjusting EM wave properties such as phase, direction, and amplitude. Through this precise control, RISs perform EM functionalities such as beam steering to redirect signals to specific targets, absorption to reduce interference or to harvest energy, and diffusion to spread signals \cite{Commag,zeris,RRS}. As a result, RISs can reshape the propagation environment in real time, improving network coverage and capacity while meeting the stringent quality of service (QoS) requirements expected in 6G applications.

The ability to precisely manipulate EM waves paves the way for futuristic services that require highly accurate EM control.  For example, the authors of \cite{XRRF} introduced the XR-RF concept, where RISs direct precise wavefronts to receivers to support ultra-low-latency extended reality applications by translating wavefronts into graphics, while the authors of \cite{Guo2024} proposed using RISs for advanced physical layer security by dynamically generating and controlling secure communication zones to prevent unauthorized eavesdropping. Achieving these advanced functionalities, however, requires a large number of reflecting elements, as each element contributes to a more refined control over the wave characteristics \cite{RISsurvey,PWE1}. However, when such precision is required, even small inaccuracies in RIS configuration can result in significant QoS degradation, thus undermining the capabilities of PWEs. Thus, robust RIS configuration schemes are essential to realize futuristic PWE-based services.

Given the need for robust RIS configuration schemes, precise receiver localization has emerged as an effective strategy for determining optimal phase shifts and reflection angles, allowing RISs to adapt their EM functionalities to specific user positions and orientations \cite{PWE1,locationdriven}. In particular, unlike traditional channel state information (CSI)-based methods that rely on iterative channel estimation with high computational overhead, localization-driven configuration reduces complexity and avoids the challenges associated with frequent CSI updates in dynamic environments \cite{locationdriven}. In addition, localization-driven approaches provide a more stable and predictable basis for RIS configuration because the spatial information changes relatively slowly over time, reducing the need for frequent recalibration. Furthermore, the resulting reduction in computational complexity improves system responsiveness, making location-driven configurations particularly well suited for latency-sensitive applications of 6G networks \cite{locationdriven}. To this end, achieving the precision required for accurate wavefront manipulation requires a locate-and-then-configure (LATC) approach capable of continuous user localization and real-time PWE adaptation. 

\subsection{State of the Art}
Building on the central role of RISs in PWEs for precise EM wave manipulation, accurate LATC schemes are critical. This challenge arises from the dual requirement for RISs to precisely manage wave propagation and to dynamically adjust their functionality according to the position of the receiver, which is often apriori unknown. One approach to address this challenge is to implement an LATC strategy that leverages the existing RIS infrastructure to localize PWE users \cite{localizationarxigos, Fascista2022Diff,dardari2022sweep2,Jiang2023sweep,Gridding1,Gridding2,Gridding3}. In particular, by dynamically adjusting wave properties such as phase, direction, and amplitude, RISs can potentially achieve high localization accuracy through tailored EM functionalities. However, in scenarios where the receiver's position is unknown, RISs must rely on location-agnostic functionalities, requiring specially designed techniques that can perform localization without prior positional knowledge. In this direction, the authors of \cite{Fascista2022Diff} introduced an iterative approach that initiates localization with a broad diffused RIS beam pattern progressively narrows over time as the receiver calculates signal to noise ratio (SNR) values from received signals and reports them back to the base station for position refinement. While this iterative feedback can improve localization accuracy, its reliance on SNR data introduces latency, limiting its effectiveness for applications requiring rapid PWE adjustment.

To extend location-agnostic functionalities, the authors of \cite{dardari2022sweep2} and \cite{Jiang2023sweep} proposed a beam-sweeping approach in which the induced RIS phase shifts are circularly varied to create a directional beam that systematically scans the environment. Specifically, in \cite{dardari2022sweep2}  beam-sweeping with OFDM pilot symbols that provide distinct frequency tones which the receiver measures and reports to the base station is implemented, thus refining localization accuracy with spatial feedback, while in \cite{Jiang2023sweep}, a two-stage beam-sweeping process is applied that starts with broad sweeps that progressively focus as the receiver feeds back signal quality measurements. However, both beam-sweeping approaches rely on iterative feedback loops, resulting in increased latency and computational complexity, which limits their suitability for PWEs that require fast and precise adaptation. To address some of these challenges, the authors of \cite{Gridding1,Gridding2, Gridding3} presented a gridding strategy that divides the service area into discrete grids, each associated with unique received signal strength (RSS) values, enabling a straightforward localization method by correlating the position of the receiver with the RSS value. However, the effectiveness of this method depends on the precise power allocation and unique RSS values across the grids, which requires extensive RIS coordination and frequent recalibration to account for environmental changes, adding complexity. Thus, while RIS-based localization approaches can assist in RIS configuration, they are susceptible to inaccuracies arising from hardware imperfections in the RIS circuitry, where such imperfections can lead to deviations in the induced phase shifts, undermining the localization precision. Consequently, such approaches may not represent the most favorable solution for achieving high-accuracy localization within PWEs.

Despite the promising capabilities of RIS-based localization, stand-alone systems dedicated solely to positioning can be instrumental in adapting PWEs to user-specific needs. In this context, the use of light-emitting diodes (LEDs) as optical anchors represents a strategic solution for high-precision localization that offers several distinct advantages. In particular, unlike RF-based approaches that can be limited by multipath effects, optical anchors can use their predominantly line-of-sight (LoS) channel to provide accurate location data through their multiple inherent degrees of freedom, such as signal intensity and angular characteristics, which are not significantly affected by scattering \cite{ghassemlooy2019optical}. In addition, the integration of LED-based systems into existing lighting infrastructures supports compatibility with modern environments and scalable localization. In this direction, the authors of \cite{kim_indoor_2013} introduced a trilateration technique based on optical RSS measurements, achieving a positioning accuracy of 2.4 cm at a fixed height. However, this two-dimensional positioning approach lacks the spatial accuracy required for precise RIS configuration, reducing its impact on effective EM functionalities. More recently, the authors of \cite{chaudhary_indoor_2021} proposed a visible light positioning system with tilted transmitters and accounting for non-LoS components leading to centimeter-level accuracy. Similarly, the authors of \cite{Wei2023ILAC} developed a federated learning model that uses location data to optimize network performance, demonstrating how accurate localization can support communication efficiency. However, both systems \cite{chaudhary_indoor_2021, Wei2023ILAC} assume a fixed receiver orientation, which limits positioning flexibility in PWEs. To address this, the authors of \cite{yasir_indoor_2016} used an accelerometer with tilted receivers to estimate 
the user orientation and exploited this measurement to achieve 6 cm positioning accuracy in three-dimensional space. However, while effective, this approach relies on accelerometer-equipped receivers, potentially limiting its general applicability when specialized equipment is not available.

Overall, existing work has demonstrated that optical-based positioning offers promising localization accuracy and robustness. However, a common limitation of these solutions is the assumption that the receivers have a static and known orientation, which restricts their ability to provide consistently accurate localization across varying receiver orientations and thus limits their suitability for PWEs. To address the orientation challenge, various machine learning-based methods, such as those in \cite{Haas2021ML, Haas2024ML}, have introduced joint position and orientation estimation through fingerprinting techniques. While effective, these approaches are typically black-box models with high computational complexity that provide limited explainability and guidance for developing LED placement strategies, and are often highly specific to the environment in which they were trained. Finally, conventional LED placements, such as ceiling LEDs, limit the range of positions and orientations that can be accurately identified, highlighting the need for an optimized deployment strategy that enables robust and accurate positioning for any user orientation. 

\subsection{Motivation and Contribution}
To the best of the authors' knowledge, no existing work has proposed a comprehensive LED deployment strategy aimed at achieving accurate user localization across arbitrary receiver orientations. Moreover, the potential to leverage optical signals for optimizing RIS configuration remains largely unexplored in the literature, opening a novel direction for precise localization within PWEs. Finally, by exploiting the synergy between RIS deployment and LED placement, the overall deployment strategy can be further refined, leading to significant improvements in PWE performance.

Motivated by the above, in this paper, we introduce an LED-based RIS-assisted PWE architecture and propose a novel LED deployment strategy to facilitate highly accurate localization based on optical RSS measurements for receivers with arbitrary orientations, allowing real-time updates of RIS configurations. Specifically, our contributions are as follows:
\begin{itemize}
    \item We present an innovative optical wireless localization scheme based on a novel LED deployment strategy for PWEs that combines conventional ceiling LEDs with \textbf{\textit{Light-emitting RISs (LeRISs)}}, i.e., RISs equipped with LEDs. Specifically, the LeRIS design enables the RIS to perform dual functions, simultaneously supporting RF wave manipulation and emitting optical signals to improve localization by introducing additional optical signal paths. By dynamically selecting the optimal LEDs from both the ceiling and the LeRIS, the PWE achieves robust localization with high spatial resolution and precise positioning across different user locations and orientations.
    \item We derive a system of closed-form equations for the coordinates of the receiver regardless of its orientation. In addition, we establish rigorous geometric constraints to ensure the solvability and uniqueness of the system of equations, providing essential guidance for LED placement within the PWE and on the LeRIS.
    \item We analyze the localization error of the proposed optical RSS-based localization method for various receiver orientations, assessing the influence of Lambertian emission order and non-Los interference on localization accuracy.
    \item To evaluate the performance of the proposed LATC scheme, extensive simulations are performed, focusing on the spectral efficiency in a communication scenario where a PWE user is served by a RIS performing beam steering . Our results show that the proposed LATC scheme can be effectively used in PWEs consisting of RISs with a large number of reflecting elements, as it allows for improved beam control. Furthermore, our results reveal that the proposed LATC scheme can compensate for RIS hardware imperfections and uncertainties in the RIS position due to the high localization accuracy achieved, underscoring its practical applicability in PWEs.
\end{itemize}
\subsection{Structure of Paper}\label{related}
The rest of the paper is organized as follows. The system model is described in Section \ref{SM}, while the equations for localization-driven beam steering are presented in Section \ref{S3}. Furthermore, the mathematical analysis of the proposed localization scheme is provided in Section \ref{S4}, and simulation results are presented in Section \ref{S5}. Finally, Section \ref{S6} concludes the paper.

\begin{figure*}[h]
\centering
\includegraphics[trim=0 0 0 0, clip, width=0.95\textwidth]{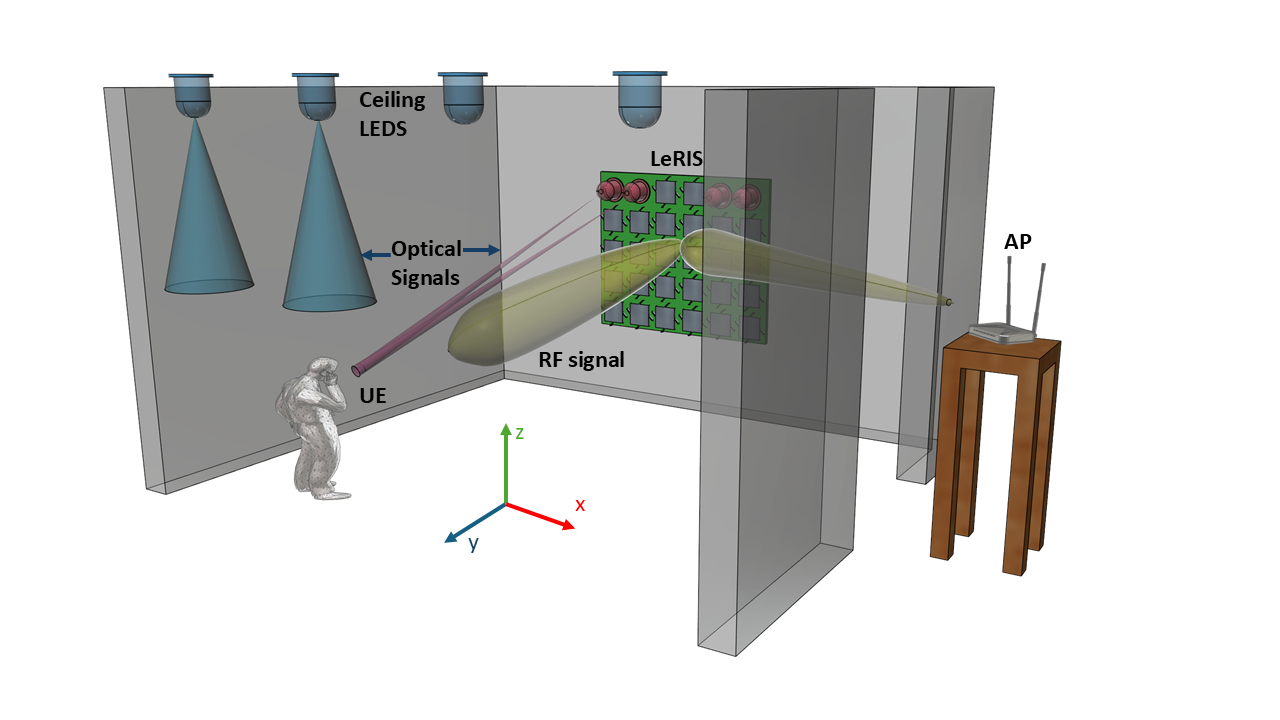}
\vspace{-2mm}
\caption{Programmable wireless environment with ceiling LEDs and Light-emitting RIS (LeRIS).} 
\label{fig:sm}
\vspace{-7mm}
\end{figure*}

\section{System Model}\label{SM}
In this section, we introduce the concept of an LED-empowered PWE using ceiling-mounted LEDs and LeRISs to optimize both the localization and communication processes. We then focus on the details of how precise localization of the PWE user is achieved using optical signals from ceiling-mounted LEDs and LeRISs, and how this information is used to ensure high-performance PWE service delivery.
\subsection{LED-empowered PWE}


We consider a PWE-assisted communication scenario, as shown in Fig. \ref{fig:sm}, where a single-antenna access point (AP) transmits within a PWE to serve a single-antenna user equipment (UE) that requires extremely high data rates but cannot establish a direct link to the AP due to blockage. To meet this requirement, the PWE orchestrator configures a RIS deployed on a wall in the PWE, assumed to be in the $xz$ plane and centered at the point $O \equiv (x_c, 0, z_c)$. We use an LED-based LATC scheme to direct the AP's transmission towards the UE with high directivity, ensuring that the UE receives the required data rate over a robust wireless link. Furthermore, multiple LEDs are strategically deployed throughout the PWE to provide an optical-based localization framework tailored for indoor scenarios. In more detail, the deployed LEDs operate at different frequencies in the infrared (IR) spectrum to avoid interference with visible light. However, in addition to conventional ceiling LEDs, multiple LEDs are placed directly on the available RISs, giving rise to the LeRIS concept. Thus, by leveraging conventional optical equipment such as cameras or photodiodes (PDs), this proposed PWE architecture can dynamically track the position of the UE and seamlessly adapt the RIS configuration to ensure high-performance PWE services.

\subsection{LATC Scheme for LED-powered PWEs}
As discussed above, accurate localization is crucial for precise RIS configuration, which is essential for ensuring optimal PWE performance. In the proposed approach, an LED-based LATC scheme minimizes configuration errors by leveraging the strong LoS characteristics of optical channels provided by both ceiling-mounted LEDs and LeRISs. This process begins with the AP broadcasting a beacon signal to announce the availability of the PWE, while each LED emits signals at a unique frequency within the IR spectrum. These distinct frequencies enable the UE to differentiate individual LEDs and calculate its position. Once the position is determined, the UE reports its location information to the AP through RF transmission, a process facilitated by configuring the RIS to operate as a randomly reconfigurable surface that scatters the transmitted RF signals randomly \cite{RRS}. This random scattering enhances the probability of signal reception at the AP, which is further supported by the established LoS channel between the AP and the RIS. Notably, the location data transmitted by the UE is minimal in size, reducing the need for a high-capacity communication channel. Upon receiving the reported location data, the AP selects the optimal RIS configuration to provide the required PWE service to ensure efficient, seamless, and highly accurate service delivery within the PWE.



\section{Localization-Driven Beam Steering}\label{S3}

By obtaining the UE position using an LATC scheme, the next step is to configure the RIS to steer the wavefront towards the UE. For example, in scenarios where high data rates are required, accurate steering of the transmitted signals becomes critical to maximize the received SNR, thereby optimizing the achievable spectral efficiency for the specific UE position. Thus, by utilizing the knowledge of the UE position and considering a RIS comprising $M \times N$ square reflecting elements with side length $D$, the achievable gain in the direction $(\theta, \phi)$, when performing beam steering, is given by \cite{cui}
\begin{equation}
    G \left(\theta, \phi \right) = {\eta_{\mathrm{eff}}}  \dfrac{4 \pi \lvert F\left(\theta, \phi \right)\rvert ^2}{\int_{0}^{2\pi} \int_{0}^{\frac{\pi}{2}}\lvert F\left(\theta, \phi \right)\rvert ^2 \mathrm{sin}\left(\theta\right) d\theta d\phi},
\end{equation}
where $\theta$ and $\phi$ are the elevation and azimuth angles, respectively, $\eta_{\mathrm{eff}}$ represents the RIS efficiency, and $F\left(\theta, \phi \right)$ is the far-field function scattered by the RIS, expressed as \cite{Abadal}
\begin{equation}\label{farfield}
    F\left(\theta, \phi \right) = \sum_{m=1}^{M} \sum_{n=1}^{N} e^{j \left(k_0 \zeta_{mn}\left(\theta,\phi\right) + \omega_{mn} + \Phi_{mn}\right)}.
\end{equation}
Here, $k_0= \frac{2 \pi}{\lambda}$ represents the wavenumber, where $\lambda$ denotes the wavelength, and $\zeta_{mn}\left(\theta,\phi\right)$ is the relative phase shift of the elements with respect to the radiation pattern coordinates, which is given as \cite{Abadal}
\begin{equation}
\begin{split}
    \zeta_{mn}\left(\theta,\phi\right) & \!=\!  D \mathrm{sin}\left(\theta\right) \left[ \left(m-\frac{1}{2}\right)\mathrm{cos}\left(\phi\right) + \left(n-\frac{1}{2}\right)\mathrm{sin}\left(\phi\right)\right] \\
    & \qquad + x_c \mathrm{sin}\left(\theta\right) \mathrm{cos}\left(\phi\right) + z_c\cos(\theta),
\end{split}
\end{equation}
Moreover, $\omega_{mn}$ is the phase impinging upon the $(m,n)$-th reflecting element and is given by \cite{Abadal}
\begin{equation}
\begin{split}
    &\omega_{mn} = k_0 \Bigg(\left(x_{\mathrm{AP}}- D \left(m-\frac{1}{2}\right) \mathrm{sin}\left(\theta\right)\mathrm{cos}\left(\phi\right) - x_c\right)^2 \\& + y^2_{\mathrm{AP}}+\left(z_{\mathrm{AP}}- D \left(m-\frac{1}{2}\right) \mathrm{sin}\left(\theta\right)\mathrm{sin}\left(\phi\right) - z_c\right)^2 \Bigg)^{1/2} ,
\end{split}
\end{equation}
where $x_{\mathrm{AP}}$, $y_{\mathrm{AP}}$, and $z_{\mathrm{AP}}$ are the AP coordinates. Furthermore, in \eqref{farfield}, $\Phi_{mn}$ is the phase shift induced by the $(m,n)$-th element, which depends on the desired RIS functionality \cite{cui}. Thus, according to the beam steering  principles described in \cite{Abadal}, $\Phi_{mn}$ is selected to enable steering in a specific direction using estimated angles obtained by the localization process. Thus, for an estimated direction (\(\hat{\theta}_r\), \(\hat{\phi}_r\)), the phase shift \(\Phi_{mn}\) induced by the $(m,n)$-th reflecting element is given by
\begin{equation}\label{phi_est}
\begin{split}
     \Phi_{mn}  = &- k_0 D \left(m \cos(\hat{\phi}_r) \sin(\hat{\theta}_r) + n \sin(\hat{\phi}_r)\sin(\hat{\theta}_r) \right) 
     \\& - \omega_{mn} - \Tilde{\theta}_{\mathrm{hw}},
\end{split}
\end{equation}
where $\Tilde{\theta}_{\mathrm{hw}}$ is a Von Mises distributed random variable with mean $\mu_{\mathrm{hw}}=0$ and concentration parameter $\kappa_{\mathrm{hw}} \in [0, +\infty)$ that quantifies the impact of RIS hardware imperfections on $\Phi_{mn}$ \cite{Justin}.

Since the direction of the UE cannot be perfectly obtained from the localization process, due to potential uncontrollable scattering of optical paths and noise effects at the receiver's optical equipment, the estimated direction of the UE may not be perfectly accurate, which may result in incorrect RIS configuration. Specifically, given that the phase shift $\Phi_{mn}$ is set according to \eqref{phi_est} based on the estimated direction $(\hat{\theta}_r, \hat{\phi}_r)$, the gain provided by the RIS to the UE whose actual direction is $(\theta_u, \phi_u)$ is given by \eqref{gain_UE} shown at the top of the next page.
\begin{figure*}
	\begin{equation}\label{gain_UE}
    G \left(\theta_u, \phi_u \right) = {\eta_{\mathrm{eff}}}  \dfrac{4 \pi \left| \sum_{m=1}^{M} \sum_{n=1}^{N} e^{j \left(\omega_{mn} + k_0 \zeta_{mn}\left(\theta_u,\phi_u\right) + \Phi_{mn}\right)}\right| ^2}{\int_{0}^{2\pi} \int_{0}^{\frac{\pi}{2}}\left| \sum_{m=1}^{M} \sum_{n=1}^{N} e^{j \left(\omega_{mn} + k_0 \zeta_{mn}\left(\theta,\phi\right) + \Phi_{mn}\right)}\right| ^2 \mathrm{sin}\left(\theta\right) d\theta d\phi},
	\end{equation}
	\hrule
\end{figure*}
Consequently, the received signal $y$ at the UE can be modeled as
\begin{equation} \label{rec_sig}
    y =  \sqrt{l_p G_t G_r P_t A_{\mathrm{eff}}  G \left(\theta_u, \phi_u \right)}x + w_n,
\end{equation}
where $x$ is the transmitted signal, assumed to satisfy $\mathbb{E}[|x|^2]=1$ with $\mathbb{E}[\cdot]$ denoting expectation and $w_n$ is additive white Gaussian noise with zero mean and variance $\sigma_n^2$. Additionally, $G_t$ and $G_r$ are the AP and UE antenna gains, respectively. Moreover, $P_t$ is the AP transmit power, and $A_{\mathrm{eff}}$ is the RIS effective area, which is given by
\begin{equation}
    A_{\mathrm{eff}}= \frac{M N G_e \lambda^2}{4 \pi},
\end{equation}
where $G_e$ denotes the element gain, and $l_p=\left(\frac{C_0}{d_1 d_2}\right)^{2}$ captures the total path loss. Here, $C_0=\frac{\lambda^2}{\left(4 \pi d_0 \right)^2}$ represents the path loss at reference distance $d_0$, while $d_1$ and $d_2$ are the AP-RIS center and RIS center-UE distance, respectively \cite{RRS}. Thus, leveraging \eqref{rec_sig}, the achievable spectral efficiency of the system can be expressed as
\begin{equation} \label{SE}
    R = \mathrm{log}_2 \left(1 + \frac{l_p G_t G_r P_t A_{\mathrm{eff}}  G \left(\theta_u, \phi_u\right)}{\sigma_n^2} \right) .
\end{equation}
From \eqref{gain_UE} and \eqref{SE}, we observe that the gain provided by the RIS, denoted as $G\left(\theta_u, \phi_u\right)$, significantly affects the achievable rate of the UE. A critical factor affecting $G\left(\theta_u, \phi_u\right)$ are the hardware imperfections inherent to the RIS circuitry, which can prevent the reflecting elements from applying the required phase shifts \cite{Justin}. In addition, inaccuracies in the knowledge of the impinging phases, represented by $\omega_{mn}$, induced by imperfect information about the RIS position, can further reduce the gain. However, while dealing with hardware imperfections and uncertainties in the impinging phases is important, the accuracy of the localization process plays an even more critical role. In particular, without accurate knowledge of the UE location, described by angles $\theta_u$ and $\phi_u$, the RIS will steer the wavefront in the wrong direction, potentially causing a significant degradation in the QoS. This highlights the importance of developing a novel LATC scheme that can minimize the effects of stochastic phenomena in the localization process, thus ensuring robust performance in scenarios where precise RIS operation is required.

\section{Optical RSS-based LATC Scheme} \label{S4}

Due to the lower susceptibility of optical wireless systems to uncontrolled scattering, several inherent degrees of freedom of optical signals can be harnessed for localization. Among the available techniques, the RSS method emerges as an effective solution for optical-based localization due to its simplicity and minimal hardware requirements \cite{Zafari2019HW}. Specifically, RSS-based localization utilizes the attenuation of optical signal strength over distance, and by correlating these values with known points in space, the system can estimate the distance between an LED and the UE. Moreover, considering that the attenuation of optical signals in indoor environments follows a reliable and predictable pattern due to the strong LoS paths, optical RSS measurements can be particularly robust for localization. As a result, the integration of optical RSS into the proposed LATC scheme enables effective RIS configuration to deliver the required accuracy in PWE services.

\subsection{System Model for Optical Localization}

Leveraging the advantages of optical RSS-based localization, the proposed LATC scheme uses optical signals emitted by multiple LEDs which are located either on the ceiling or on the LeRISs. While the environment may contain multiple LEDs, the UE receives optical signals exclusively from the subset of LEDs denoted by $\mathcal{I} = \{1, \dots, I \}$, which includes only those LEDs with which the UE shares a LoS link. Thus, the channel gain from the $i$-th LED to the UE, where $i \in \mathcal{I}$, is modeled by
\begin{equation}
    h_{i} = h_{i,\mathrm{L}} + h_{i,\mathrm{N}},
\end{equation}
where $h_{i,\mathrm{L}}$ represents the LoS component of the channel, while $h_{i,\mathrm{N}}$ accounts for the non-LoS component due to uncontrollable scattering effects. More specifically, according to \cite{Bozanis2022}, $h_{i,\mathrm{L}}$ is expressed as
\begin{equation}\label{h_los}
h_{i,\mathrm{L}}=\frac{(m_l+1)A_{\mathrm{PD}} T_{f}}{2\pi d^2_{l,i}} \cos^{m_l}(\psi_{d,i})g_{c}(\psi_{a,i})\cos(\psi_{a,i}),
\end{equation}
where $m_l$ is the Lambertian emission order, which indicates how directional the optical beam is. Additionally, $\Psi_{1/2} = \mathrm{cos}^{-1} \left( e^{-\frac{m_l}{\ln(2)}} \right)$ is the half-power beamwidth of the LED, $d_{l,i}$ is the Euclidean distance between the $i$-th LED and the UE, and $\psi_{d,i}$ and $\psi_{a,i}$ are the departure and arrival angles, respectively. Furthermore, parameter $T_{f}$ denotes the gain of the optical filter, $A_{\mathrm{PD}}$ is the area of the photodiode (PD), and $g_{c}(\psi_{a,i})$ is the optical concentrator gain of the PD, which is given by \cite{Bozanis2022}
\begin{equation}\label{16}
g_{c}(\psi_{a,i})=\left\{
\begin{array}{ll}
\frac{n_{c}^2}{\sin^2(\Psi_{\mathrm{max}})}, & 0\leq \psi_{a,i} \leq \mathrm{\Psi_{\mathrm{max}}} \\
0, & \psi_{a,i}>\mathrm{\Psi_{\mathrm{max}}},
\end{array}
\right.
\end{equation}
where $n_{c}$ is the refractive index of the PD, quantifying how much light is deflected when it enters the PD, and $\mathrm{\Psi_{\mathrm{max}}}$ is half the field of view (FoV) of the PD. In addition to the LoS component, $h_{i,\mathrm{N}}$ is affected by reflections from objects and walls, which cannot be modeled uniformly due to their different shapes and material properties. Therefore, to accurately capture optical scattering effects, reflective surfaces are divided into smaller sections that follow the Lambertian reflectance model, which assumes that light is diffusely reflected with intensity dependent on the angle of incidence \cite{Hass2018NLoS}. This approach accounts for both angular diversity and different reflective properties in different sections of objects, improving the accuracy of the non-LoS channel model. Thus, by focusing on first-order reflections, since subsequent reflections experience significant signal degradation and become negligible \cite{Hass2018NLoS}, $h_{i}^{\mathrm{N}}$ can be calculated by summing all individual reflection paths as follows
\begin{equation}
    h_{i,\mathrm{N}} = \sum_{j=1}^Jh_{ij}^{\mathrm{RP}},
\end{equation}
where $J$ is the total number of non-LoS paths, while $h_{ij}^{\mathrm{RP}}$ is the $j$-th first-order reflection path from the $i$-th LED and is given by
\begin{equation}
\begin{split}
    h^{\mathrm{RP}}_{i,j} & = \int_{\mathcal{S}}   \frac{\rho_r (m_l+1)A_{\mathrm{PD}} T_{f}}{2\pi d^2_{r,1}d^2_{r,2}}   \cos^{m}(\phi) \\ 
    & \qquad \times \cos(\psi_1) \cos(\psi_2)g_{c}(\psi) 
   \cos(\psi)\mathrm{d}\mathcal{S}
\end{split}
\end{equation}
where $d_{r,1}$ is the distance between the LED and the reflecting object, $d_{r,2}$ is the distance from the reflecting object to the UE, $\rho_r$ denotes the reflectance of the object, and $\psi_1$ and $\psi_2$ are the angles of incidence on the object and exit from the object to the UE, respectively. Here, $\mathcal{S}$ is the total reflective surface to account for all the infinitesimal contributions of the first-order reflection path.

\subsection{RSS-based Positioning}

To determine the position $U = (x_u, y_u, z_u)$ of the UE in the PWE, the RSSs from multiple LEDs can be used because the signal strength attenuation can provide valuable information about the distance. Specifically, the power $P_{r,i}$ measured at the PD of the UE from the optical signal transmitted by the $i$-th LED in the optical domain is given by
\begin{equation}\label{RSS_1}
    P_{r,i} = \underbrace{h_{i,\mathrm{L}}P_{o,i}}_{P_\mathrm{LoS},i} + \underbrace{h_{i,\mathrm{N}}P_{o,i}}_{P_\mathrm{NLoS},i} + P_{n,i},
\end{equation}
where $P_{o,i}$ is the transmitted optical power of the $i$-th LED, $P_{\mathrm{LoS},i}$ and $P_{\mathrm{NLoS},i}$ are the received powers from the LoS and non-LoS components, respectively, and $P_{n,i}$ denotes the power of the noise, which arises from thermal noise, dark current, and background radiation \cite{ghassemlooy2019optical}. 

In an ideal scenario, where the channel gain is only affected by the LoS component, based on \eqref{RSS_1} the distance between the $i$-th LED and the UE can be calculated using \eqref{h_los} and the ratio $\frac{P_{r,i}}{P_{o,i}}$. In practice, however, the received power $P_{r,i}$ is also affected by the non-LoS component and noise, which introduces errors for distance estimation. Nevertheless, in optical wireless systems, the LoS component generally dominates over the non-LoS and noise components \cite{ghassemlooy2019optical}, which allows us to neglect the $P_\mathrm{NLoS}$ and $P_n$ terms in \eqref{RSS_1}, thus providing a good estimate of the channel gain as 
\begin{equation}\label{h_est} \hat{h}_{i,\mathrm{L}} = \frac{P_{r,i}}{P_{o,i}}. 
\end{equation}
By assuming, without loss of generality, that the transmission direction of optical signals is considered to be the positive y-axis, it means that the LEDs are placed on a plane parallel to $xz$, i.e., $y=y_i$. Thus, the $\psi_{d,i}$ from the $i$-th LED located at $(x_i,y_i,z_i)$ to the UE can be expressed as 
\begin{equation}\label{AoD}
    \psi_{d,i} = \cos^{-1}{\left(\frac{y_u - y_i}{d_{l,i}}\right)}.
\end{equation}
 It should be highlighted that this formulation holds for LEDs on other planes as well, with the appropriate coordinates substituted in the numerator of \eqref{AoD} to reflect their respective direction of transmission.
Therefore, by substituting \eqref{h_est} and \eqref{AoD} in \eqref{h_los} and after some algebraic manipulations, we can express the estimated distance $\hat{d}_{l,i}$ between the $i$-th LED and the UE as 
\begin{equation}
\begin{split}\label{d_l}
    \hat{d}_{l,i} & = \Bigg(\frac{P_{o,i}}{P_{r,i}}\frac{(m_l+1)A_{\mathrm{PD}}\left(\hat{y}_u-y_i\right)^{m_l}}{2\pi} \\
    & \qquad \qquad \times T_{f}g_{c}(\psi_{a,i})\cos(\psi_{a,i})\Bigg)^{1/(m_l+2)}.
    \end{split}
\end{equation}
It should be noted that in \eqref{d_l}, if $P_\mathrm{LoS,i}$ was used in place of $P_{r,i}$, the calculated distance would represent the actual value rather than an estimate. Finally, by substituting $\hat{d}_{l,i}$ with its Cartesian coordinate form and some algebraic manipulations, \eqref{d_l} can be rewritten as
\begin{equation} \label{coor}
\begin{aligned}
    &\frac{(\hat{x}_u-x_i)^2 + (\hat{y}_u-y_i)^2 + (\hat{z}_u - z_i)^2}{(\hat{y}_u-y_i)^{\frac{2m}{m_l+2}}}  \\ 
    & \qquad = \left (\frac{P_{o,i}}{P_{r,i}} \frac{(m_l+1)A_{\mathrm{PD}}}{2\pi}T_{f}g_{c}(\psi_{a,i})\cos(\psi_{a,i})\right)^{2/(m_l+2)} .
\end{aligned}
\end{equation}

Examining \eqref{coor}, it becomes clear that a system of equations is needed to determine the position of the UE. Although three variables corresponding to the coordinates of the UE are required to define its position, three equations alone are insufficient due to the non-linear nature of the system. Therefore, it has been proven in \cite{Bozanis2023Loc} that for a fixed UE orientation, a fourth equation is required to resolve these non-linearities and ensure a unique and accurate solution to the system. These four equations can be derived from the RSS measurements of four LEDs, with each LED providing an equation corresponding to its distance from the UE, as expressed in \eqref{coor}. However, the random orientation of the UE introduces additional complexity regarding the unknown angles $\psi_{a,i}$ from each LED. Specifically, with $I$ LEDs, the total number of unknowns becomes $I+3$, as each LED introduces a unique angle of incidence $\psi_i$ in addition to the unknown UE coordinates, thereby increasing the minimum number of equations needed to ensure a unique solution. Thus, it is necessary to derive additional equations to provide more information about the system variables without increasing the number of unknowns. Furthermore, the random orientation of the UE can lead to different $\psi_{a,i}$ values, which can lead to linear dependencies between the equations and make the system unsolvable. Consequently, the LEDs must be carefully placed to ensure that the equations remain linearly independent even with random $\psi_{a,i}$ values. In this context, the following lemma establishes the constraints on the LED coordinates, which depend on the UE orientation, ensuring effective localization with the minimum possible number of LEDs.
\begin{lemma}\label{l1}
The UE position can be determined using the optical RSS-based method if it has LoS links with four LEDs placed on a plane, provided that no two LEDs are equidistant from the UE when this plane is parallel to the UE plane.



\end{lemma}
\begin{IEEEproof}
The proof is presented in Appendix A.
\end{IEEEproof}

\begin{remark}
According to the proof of Lemma \ref{l1}, certain LEDs may be unsuitable for UE localization at specific UE positions and orientations. To address this, it is essential to have multiple LEDs strategically placed in different locations, ensuring that the UE can always perform accurate localization regardless of its position or orientation.
\end{remark}
\begin{remark}
    Because the UE position cannot be determined if the UE plane is parallel to a plane containing two or more LEDs and the UE is equidistant from any pair of these LEDs, it is crucial to place LEDs at multiple levels to avoid such scenarios. Therefore, by ensuring a diverse placement of LEDs, such as combining ceiling LEDs with LeRIS, the system can always select LEDs from different planes, reducing the likelihood of making the system of equations unsolvable. This guarantees that the UE plane cannot be simultaneously parallel to all the planes on which the LEDs are placed. Thus, LeRIS can play a critical role by offering additional optical paths and providing spatial diversity through LEDs placed on different planes, ensuring robust localization and improving overall PWE efficiency.
\end{remark}

By appropriately placing the LEDs according to Lemma \ref{l1}, we can obtain geometric equations that relate the LED positions to the ratio between the cosines of the angles of incidence from each LED, and thus are able to solve the system of equations. In this direction, a way to deploy them that satisfies these constraints is to place them all at the same height, i.e., $z=z_\mathrm{L}$, as shown in Fig. \ref{fig:sm}. More specifically, this arrangement satisfies the constraint of not having LEDs with the same $x$ coordinate, while it also gives the freedom to choose different distances between the LEDs on the $x$ axis, making the system more tolerant of potential equality in the distances between the LEDs and the UE. Therefore, in the following proposition, we derive the UE coordinates for the case where the LED positions coincide with the points $\mathrm{A} \equiv (x_\mathrm{A},0,z_\mathrm{L})$, $\mathrm{B} \equiv (x_\mathrm{B},0,z_\mathrm{L})$, $\mathrm{C} \equiv (x_\mathrm{C},0,z_\mathrm{L})$, and $\mathrm{D} \equiv (x_\mathrm{D},0,z_\mathrm{L})$.
\begin{proposition}
     If the optical RSS-based method is used and the LED positions coincide with the points $\mathrm{A}$, $\mathrm{B}$, $\mathrm{C}$, and $\mathrm{D}$, the estimated UE coordinates $\left(\hat{x}_u, \hat{y}_u, \hat{z}_u\right)$ can be calculated by \eqref{x}, \eqref{l}, and \eqref{z}, with $\alpha_1$, $\alpha_2$, and $\alpha_3$ given by
    \begin{align}
   \alpha_1 &= \left(\frac{{P_{r,2}}}{{P_{r,1}}}\right)^{2/\left(m_l+1\right)}, \label{geo1} \\   \alpha_2 &= \left(\frac{{P_{r,3}}}{{P_{r,1}}}\right)^{2/\left(m_l+1\right)}, \label{geo2}  \\
\alpha_3&= \left(\frac{{P_{r,4}}}{{P_{r,1}}}\right)^{2/\left(m_l+1\right)}.   \label{geo3} 
\end{align}
\end{proposition}
\begin{IEEEproof}
By expressing the incidence angles from the triangles formed between the LEDs, the UE, and the normal to the UE vector, the following expressions can be derived:
\begin{align}
   \frac{\cos(\psi_{a,1})}{\cos(\psi_{a,2})} &= \frac{\hat{d}_{l,2}}{\hat{d}_{l,1}}, \label{geo01} \\   \frac{\cos(\psi_{a,1})}{\cos(\psi_{a,3})} &= \frac{\hat{d}_{l,3}}{\hat{d}_{l,1}}, \label{geo02}  \\
\frac{\cos(\psi_{a,1})}{\cos(\psi_{a,4})} &= \frac{\hat{d}_{l,4}}{\hat{d}_{l,1}}.
   \label{geo03} 
\end{align}
Thus, utilizing \eqref{d_l}, \eqref{a1}, \eqref{a2} and \eqref{a3}, equations \eqref{geo01}, \eqref{geo02}, and \eqref{geo03} for the 4 different distances can be rewritten as \eqref{geo1}, \eqref{geo2}, and \eqref{geo3}. Thus, the UE coordinates can be obtained, which concludes the proof. 
\end{IEEEproof}
\begin{remark}
For the case where all 4 LEDs utilized for localization, are placed upon a plane that is parallel to the UE plane, i.e., the angles of incidence are equal to the angles of departure, i.e., $\psi_{d,i} = \psi_{a,i}, \forall i \in \{1,2,3,4\} $, it can be proved that \eqref{a1}, \eqref{a2}, and \eqref{a3} do not remain functions of these angles. Therefore, the solution of the problem can be derived directly from \eqref{z}, \eqref{x}, and \eqref{l} without using equations \eqref{geo01}, \eqref{geo02}, and \eqref{geo03}.
\end{remark}

Considering the existence of non-LoS components and noise, it becomes clear that the UE location cannot be acquired perfectly, because the estimated UE coordinates have been calculated based on the assumption that the optical channel is only affected by the LoS component. In this direction, in the following proposition, we provide the localization error of the optical RSS-based localization method.
\begin{proposition}
The localization error of the optical RSS-based method can be expressed as
    
\begin{equation}\label{error_dist}
 {\Delta d} = d_{l}\left(1 - \left(\frac{1}{1 + K_\mathrm{opt}^{-1} + \frac{P_{n}}{P_\mathrm{LoS}}}\right)^{1/(m_l+2)} \right),
\end{equation}
where $K_\mathrm{opt} = \frac{P_\mathrm{LoS}}{P_\mathrm{NLoS}}$ is the ratio of the power of the LoS and non-LoS components of the optical channel. 

\end{proposition}
\begin{IEEEproof}
The localization error between the actual and estimated distance, utilizing \eqref{d_l}, can be expressed as
\begin{equation}\label{37}
\begin{split}
    {\mathrm{\Delta}d}  = & d_{l} - \hat{d}_{l} \\
     = &\Bigg( \frac{P_{o}(m_l+1)A_{\mathrm{PD}}\hat{y}_u^{m}}{2\pi}T_{f} g_{c}(\psi_{a})\cos(\psi_{a})\Bigg)^{1/(m_l+2)}  \\
     \quad \times & 
    \left(\frac{1}{P_{\mathrm{LoS}}}-\frac{1}{P_{r}}\right)^{1/(m_l+2)}.\\
\end{split}
\end{equation}
Finally, after some algebraic manipulations in \eqref{37}, then \eqref{error_dist} can be derived, which completes the proof. 
\end{IEEEproof}

\subsection{LATC-based RIS Configuration}
After the UE position has been extracted, it is reported back to the AP and used as input for the RIS configuration process. However, since the position of the UE is not perfectly estimated, the induced phase shifts of the RIS are affected, which also influences the achievable spectral efficiency of the system. Therefore, in the following proposition, we provide the formula that describes the induced phase shift from each reflecting element for the case where the RIS performs beam steering and the UE location is acquired through the optical RSS-based method.
\begin{proposition}
The induced phase shift $\hat{\Phi}_{mn}$ from the $(m, n)$-th reflecting element for the case where the RIS performs the beam steering  functionality and the UE location is acquired through the optical RSS-based method is given by \eqref{phi_est}, where $\hat{\theta}_r$ and $\hat{\phi}_r$ are given, respectively, as
\begin{equation}\label{theta_r}
    \hat{\theta}_r = \tan^{-1}\left(\frac{\hat{z}_u-z_c}{\sqrt{(\hat{x}_u-x_c)^2+\hat{y}_u^2}}\right),
\end{equation}
and
\begin{equation}\label{phi_r}
   \hat{\phi}_r =  \tan^{-1}\left(\frac{\hat{y}_u}{\hat{x}_u-x_c}\right).
\end{equation}
\end{proposition}
\begin{IEEEproof}
By observing the triangles between $O$, $\hat{U}$ and the projection of $\hat{U}$ onto the $xy$ plane, and the positive $x$ axis, the elevation and azimuth angles between the points $O$ and $\hat{U}$ can be calculated, respectively, by \eqref{theta_r} and \eqref{phi_r}, which concludes the proof.
\end{IEEEproof}

\section{Simulation Results}\label{S5}
In this section, we evaluate the performance of the proposed LATC scheme through simulations, assuming a communication scenario with a single-antenna AP and a single-antenna UE within a PWE equipped with 4 LEDs on the ceiling and a square-shaped LeRIS (i.e., $M = N$) with 4 LEDs mounted on it, as shown in Fig. \ref{fig:sm}. Specifically, the positions of the nodes and the dimensions of the PWE are given in Table \ref{PWE_char}, while the selected values for the parameters of both communication and localization systems are given in Table II. To ensure an accurate evaluation, we perform a Monte Carlo simulation with $10^5$ realizations of the UE location and orientation, where in each realization the coordinates of the UE are randomly chosen within $x_u \in [0,10\,\mathrm{m}]$, $y_u \in [0.5,5\, \mathrm{m}]$, and $z_u \in [0.5,2.5\,\mathrm{m}]$, while the UE's elevation and azimuth angles vary within $\theta_\mathrm{UE} \in \left[-\frac{\pi}{2}, \frac{\pi}{2}\right]$ and $\phi_\mathrm{UE} \in \left[-\frac{\pi}{2}, \frac{\pi}{2}\right]$, respectively. In addition, the average value of $K_{\text{opt},i} $, which captures the ratio of LoS to non-LoS components based on reflections from surrounding walls and objects, is set to the same value $K_{\mathrm{PWE}}$, unless stated otherwise, to represent a value typically found in PWEs. Furthermore, we assume that the AP has a distance $d_1=20$ m from the LeRIS and that the UE communicates with the AP using the LeRIS, which steers the waves transmitted by the AP to the receiver using beam steering. Finally, unless otherwise stated, we assume a Lambertian emission order $m_l=2$, and both AP and UE are equipped with isotropic antennas, i.e., $G_t=G_r=1$.

\begin{table}[]
	\renewcommand{\arraystretch}{1.25}
	\caption{\textsc{PWE Characteristics}}
	\label{PWE_char}
	\centering
	\begin{tabular}{ll}
		\hline
		\bfseries Parameter & \bfseries Value \\
		\hline\hline
		Ceiling LED	1	 	    & $\left(3.5, 5, 3\right)$ 			        \\
		Ceiling LED	2	 	    & $\left(4.5, 5, 3\right)$ 			        \\
  	Ceiling LED	3	 	    & $\left(5.5, 5, 3\right)$ 			        \\
        Ceiling LED	4	 	    & $\left(6.5, 5, 3\right)$ 			        \\
        LeRIS Center 	    & $\left(5, 0, 1.5\right)$ 			        \\
        LeRIS LED 1	 	    & $\left(4.1, 0, 1.5\right)$ 			        \\
        LeRIS LED 2	 	    & $\left(4.7, 0, 1.5\right)$ 
        \\
        LeRIS LED 3	 	    & $\left(5.3, 0, 1.5\right)$
        \\
        LeRIS LED 4 	    & $\left(5.9, 0, 1.5\right)$ 
        \\
        PWE Dimensions 	    & $10\,\mathrm{m} \times 10\,\mathrm{m} \times3\,\mathrm{m} $ 
        \\
		\hline
	\end{tabular}
\end{table}

\begin{table}[]
	\renewcommand{\arraystretch}{1.25}
	\caption{\textsc{Simulation Result Parameters}}
	\label{Sim_res}
	\centering
	\begin{tabular}{ll}
		\hline
		\bfseries Parameter & \bfseries Value \\
		\hline\hline
		Wavelength 	    & $\lambda=\frac{1}{8}$ m 			        
        \\
        RIS efficiency 	    & $n_{\mathrm{eff}}=100 \%$ 			        
        \\
		Element Size	 	    & $D = \frac{\lambda}{2}$ m 			            \\
  	Element Gain	 	    & $G_e = 1$ 			            \\
  	Reference Distance	 	    & $d_0 = 1$ m 			       
        \\
        RF Noise power	 	    & $\sigma_n^2 = -130\, \mathrm{dB}$ 			        \\
        RF transmit power 	    & $P_t = 1$ W 			        \\
        Optical transmit power 	    & $P_o = 50$ mW 			        \\
        Optical Noise power	 	    & $P_{n,i} = 5\times 10^{-14}$ A${}^2$ 
        \\
        Optical Filter Gain 	    & $T_f = 1$
        \\
        PD Area 	    & $A_\mathrm{PD} = 1$ cm${}^2$ 
        \\
        PD refractive index 	    & $n_c = 1.5$
        \\
        Half of the FoV angle & $\Psi_\mathrm{max} = 75^{\circ}$
        \\
		\hline
	\end{tabular}
\end{table}

Fig. \ref{fig:figure1} shows the localization error $\Delta d $ as a function of the UE's rotation angle $\theta_{\mathrm{UE}}$ for a given position with $K_{\mathrm{PWE}}=100 $ and $\phi_{\mathrm{UE}}=0^\circ$, for the case where only ceiling LEDs are used and for the case where both ceiling LEDs and a LeRIS are used. Specifically, the UE plane is initially parallel to the ceiling LEDs plane at $\theta_{\mathrm{UE}} = 90^\circ$ and gradually aligns with the LeRIS plane at $\theta_{\mathrm{UE}} = 0^\circ$. As can be seen, for the case where only ceiling LEDs are used by the LATC scheme, the localization error increases sharply as $\theta_\mathrm{UE}$ exceeds $45^\circ$, reflecting their limited coverage as the UE moves beyond their optimal range. However, when the LeRIS is added, the localization error remains significantly lower as it provides additional signal strength. In addition, although a higher Lambertian order $m_l=2$ generally improves accuracy for most $\theta_\mathrm{UE}$ values due to its focused emission, within the range $[40^\circ, 50^\circ]$, the broader optical emission of $m_l = 1$ outperforms $m_l = 2$, indicating that the optimal $m_l$ value varies with UE plane orientation. Nevertheless, it can be observed that for both $m_l$ values, a LeRIS leads to more stable localization over a wide range of angles, highlighting the importance of equipping a PWE with LeRISs to mitigate signal degradation and maintain high localization accuracy as the UE orientation varies.

  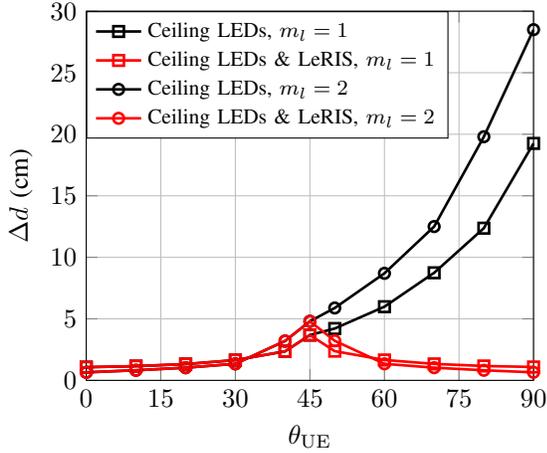
\begin{figure}
    \centering
    \begin{tikzpicture}
        \begin{axis}[
            width=0.85\linewidth,
            xlabel = {$\theta_\mathrm{UE}$},
            ylabel = {$\Delta d$ (cm)},
            ymin = 0,
            ymax = 30,
            xmin = 0,
            xmax = 90,
            ytick = {0,5,10,15,20,25,30},
            xtick = {0,15,30,...,90},
            grid = major,
		  legend entries ={{Ceiling LEDs, $m_l=1$},{Ceiling LEDs \& LeRIS, $m_l=1$},{Ceiling LEDs, $m_l=2$},{Ceiling LEDs \& LeRIS, $m_l=2$}},
            legend cell align = {left},
            legend style={font=\footnotesize},
            legend style={at={(0,1)},anchor=north west},
            ]
            \addplot[
            black,
            mark = square,
            mark repeat = 1,
            mark size = 2,
            line width = 1pt,
            style = solid,
            ]
            table {Fig2/ph1.dat};
            \addplot[
            red,
            mark = square,
            mark repeat = 1,
            mark size = 2,
            line width = 1pt,
            style = solid,
            ]
            table {Fig2/phi2.dat};
            \addplot[
            black,
            mark = o,
            mark repeat = 1,
            mark size = 2,
            line width = 1pt,
            style = solid,
            ]
            table {Fig2/m2LeRIS2_1.dat};
            \addplot[
            red,
            mark = o,
            mark repeat = 1,
            mark size = 2,
            line width = 1pt,
            style = solid,
            ]
            table {Fig2/m2LeRIS.dat};
        \end{axis}
    \end{tikzpicture}
    \caption{Localization error $\Delta d$ for varying azimuth angle $\theta_\mathrm{UE}$.}
    \label{fig:figure1}
\end{figure}


In Fig. \ref{fig:subfig2a} and Fig. \ref{fig:subfig2}, the localization error $\Delta d$ is plotted as a function of $K_{\mathrm{PWE}}$ for three Lambertian emission orders, $m_l=1$, $m_l=2$, and $m_l=5$, comparing scenarios with only ceiling LEDs and with both ceiling LEDs and a LeRIS. As $K_{\mathrm{PWE}}$ increases, $\Delta d$ decreases, highlighting the improvement in accuracy with greater LoS dominance. Furthermore, when using only ceiling LEDs, $\Delta d$ remains relatively high, especially at lower $K_{\mathrm{PWE}}$ values, while including LeRIS significantly reduces the error at all $K_{\mathrm{PWE}}$ levels. Furthermore, higher Lambertian orders, such as $m_l=5$, tend to produce higher localization errors at lower $K_{\mathrm{PWE}}$ due to their narrower emission patterns, which limit coverage for certain orientations. On the contrary, $m_l=2$ provides a balanced performance, achieving lower localization errors without the limitations of an omnidirectional pattern, i.e., $m_l=1$, or a highly directive pattern, i.e., $m_l=5$. However, as $K_{\mathrm{PWE}}$ continues to increase, this performance gap narrows, indicating that LoS clarity becomes the dominant factor in determining localization accuracy at high $K_{\mathrm{PWE}}$ values. To this end, Fig. \ref{fig:subfig2a} and Fig. \ref{fig:subfig2} highlight the role of LeRIS in improving localization performance, as well as the impact of $m_l$ on the LATC scheme.


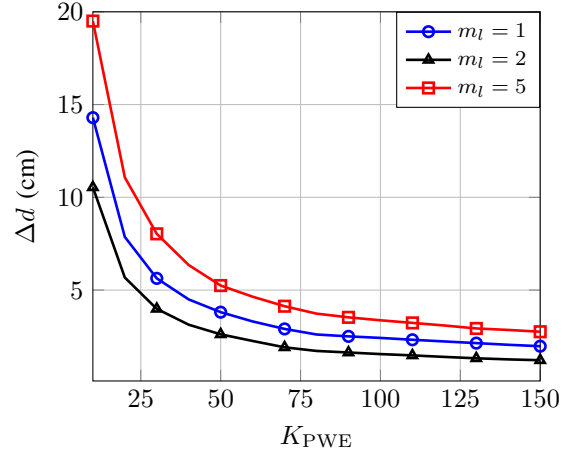
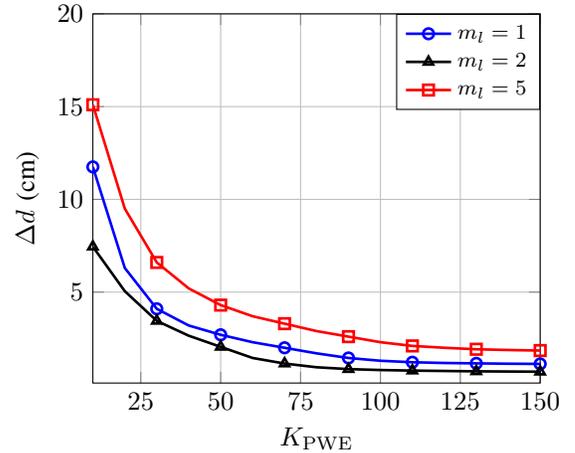
\begin{figure}
    \centering
    \begin{subfigure}{\linewidth}
        \centering
        \begin{tikzpicture}
            \begin{axis}[
                 width=0.85\linewidth,
                xlabel = {$K_{\mathrm{PWE}}$},
                ylabel = {$\Delta d$ (cm)},
                ymin = 0.1,
                ymax = 20,
                xmin = 10,
                xmax = 150,
                xtick = {25,50,...,150},
                grid = major,
                legend entries ={{$m_l=1$},{$m_l=2$},{$m_l=5$}},
                legend cell align = {left},
                legend style={font=\footnotesize},
                legend style={at={(1,1)},anchor=north east},
            ]
                \addplot[
                blue,
                mark = o,
                mark repeat = 2,
                mark size = 2,
                line width = 1pt,
                style = solid,
                ]
                table {Fig1/m1.dat};
                \addplot[
                black,
                mark = triangle,
                mark repeat = 2,
                mark size = 2,
                line width = 1pt,
                style = solid,
                ]
                table {Fig1/m2.dat};
                \addplot[
                red,
                mark = square,
                mark repeat = 2,
                mark size = 2,
                line width = 1pt,
                style = solid,
                ]
                table {Fig1/m5.dat};
            \end{axis}
        \end{tikzpicture}
        \caption{Ceiling LEDs.}
        \label{fig:subfig2a}
    \end{subfigure}
    \begin{subfigure}{\linewidth}
        \centering
        \begin{tikzpicture}
            \begin{axis}[
                width=0.85\linewidth,
                xlabel = {$K_{\mathrm{PWE}}$},
                ylabel = {$\Delta d$ (cm)},
                ymin = 0.1,
                ymax = 20,
                xmin = 10,
                xmax = 150,
                xtick = {25,50,...,150},
                grid = major,
                legend entries ={{$m_l=1$},{$m_l=2$},{$m_l=5$}},
                legend cell align = {left},
                legend style={font=\footnotesize},
                legend style={at={(1,1)},anchor=north east},
            ]
                \addplot[
                blue,
                mark = o,
                mark repeat = 2,
                mark size = 2,
                line width = 1pt,
                style = solid,
                ]
                table {Fig1/m1ler.dat};
                \addplot[
                black,
                mark = triangle,
                mark repeat = 2,
                mark size = 2,
                line width = 1pt,
                style = solid,
                ]
                table {Fig1/m2ler.dat};
                \addplot[
                red,
                mark = square,
                mark repeat = 2,
                mark size = 2,
                line width = 1pt,
                style = solid,
                ]
                table {Fig1/m5ler.dat};
            \end{axis}
        \end{tikzpicture}
        \caption{Ceiling LEDs \& LeRIS.}
        \label{fig:subfig2}
    \end{subfigure}
    \caption{Localization error $\Delta d$ versus $K_{\mathrm{PWE}}$ for various $m_l$ values.}
    \label{fig:figure2}
\end{figure}

Fig.~\ref{fig:figure3} illustrates the spectral efficiency as a function of $ K_{\mathrm{PWE}} $ for LeRIS sizes of $N=100$, $N=625$, and $N=1600$ reflecting elements, highlighting the effect of localization accuracy for different numbers of reflecting elements. Initially, at lower $K_{\mathrm{PWE}}$ values, where localization accuracy is limited, smaller RISs achieve comparable or even higher spectral efficiency than larger RIS arrays due to their broader beam coverage that compensates for imprecise beam steering. However, as $K_{\mathrm{PWE}}$ increases, $N=1600$ shows a clear advantage, reaching nearly 25 bps/Hz by using its narrower main lobe to efficiently direct power to the UE. This trend highlights the need for accurate localization to fully exploit the beam steering capabilities of larger RISs. However, given that typical $ K_{\mathrm{PWE}} $ values in indoor environments are between 50 and 100 \cite{Hass2018NLoS}, the viability of an optical RSS-based LATC scheme is well supported to achieve the localization precision required for fully exploiting the potential of larger RISs. Consequently, Fig.~\ref{fig:figure3} underscores that accurate localization is essential for effectively supporting larger numbers of reflective elements in a PWE, and highlights the potential of LeRIS to enable reliable LATC performance due to its favorable $ K_{\mathrm{PWE}} $ range.

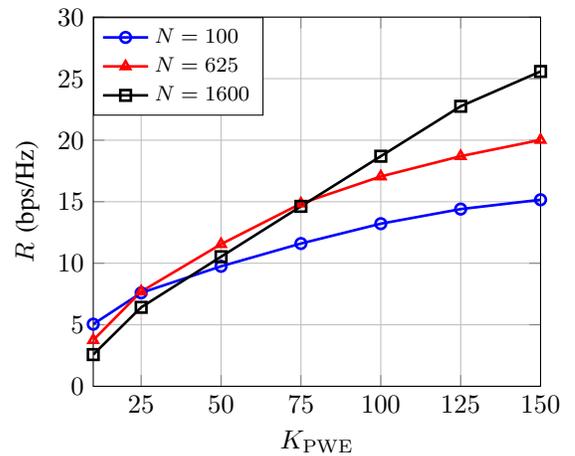
\begin{figure}
    \centering
    \begin{tikzpicture}
        \begin{axis}[
            width=0.85\linewidth,
            xlabel = {$K_{\mathrm{PWE}}$},
            ylabel = {$R$ (bps/Hz)},
            ymin = 0,
            ymax = 30,
            xmin = 10,
            xmax = 150,
            xtick = {0,25,50,...,150},
            ytick = {0,5,10,...,30},
            grid = major,
		legend entries ={{$N=100$},{$N=625$ },{$N=1600$ }},
            legend cell align = {left},
            legend style={font=\footnotesize},
            legend style={at={(0,1)},anchor=north west},
            ]
            \addplot[
            blue,
            mark = o,
            mark repeat = 1,
            mark size = 2,
            line width = 1pt,
            style = solid,
            ]
            table {Fig3/N100.dat};
            \addplot[
            red,
            mark = triangle,
            mark repeat = 1,
            mark size = 2,
            line width = 1pt,
            style = solid,
            ]
            table {Fig3/N625.dat};
            \addplot[
            black,
            mark = square,
            mark repeat = 1,
            mark size = 2,
            line width = 1pt,
            style = solid,
            ]
            table {Fig3/N1600.dat};
        \end{axis}
    \end{tikzpicture}
    \caption{Achievable spectral efficiency $R$ versus $K_{\mathrm{PWE}}$ for various LeRIS sizes.}
    \label{fig:figure3}
\end{figure}

In Fig. \ref{fig:figure4}, the effect of localization accuracy on spectral efficiency $R$ is analyzed for different levels of hardware imperfection and two LeRIS sizes, illustrating how the proposed LATC scheme mitigates imperfections and highlights the advantages of larger RISs under realistic conditions. Specifically, Fig. \ref{fig:figure4a} shows $R$ as a function of $K_{\mathrm{PWE}}$ for $N=100$ and $N=1600$, where it is evident that hardware imperfections, represented by $\kappa_{\mathrm{hw}}$, degrade performance for both sizes, with $N=1600$ being more sensitive due to its larger number of elements. This increased sensitivity arises because imperfections in $N=1600$ amplify the likelihood of incoherent element configurations, leading to pronounced side lobes and reduced main lobe gain, which diminishes $R$. Nevertheless, the LATC scheme plays a critical role in compensating for these imperfections by leveraging the localization accuracy achievable within the typical $K_{\mathrm{PWE}}$ range of indoor environments, enabling $N=1600$ to effectively utilize its larger aperture and outperform $N=100$ even in conditions of severe hardware imperfections. Further insights into the effects of hardware imperfections are provided in Figs. \ref{fig:figure4b} and \ref{fig:figure4c}, where $N=1600$ exhibits a significant reduction in main lobe magnitude and a noticeable distribution of energy into side lobes, highlighting its reliance on accurate localization to maintain performance. In contrast, $N=100$ demonstrates weaker side lobes and greater robustness to hardware imperfections, as its smaller number of elements reduces the impact of incoherent configurations. However, while $N=100$ is more resilient, its limited aperture size restricts its ability to achieve the directional gain of $N=1600$, especially when localization accuracy improves with increasing $K_{\mathrm{PWE}}$. Finally, it should be highlighted that when $\kappa_{\mathrm{hw}}=0$, $R$ becomes independent of localization accuracy because LeRIS randomly diffuses the signal over the PWE, with $N=1600$ consistently outperforming $N=100$ due to its larger size, which allows more power to impinge up it, resulting in greater signal diffusion.

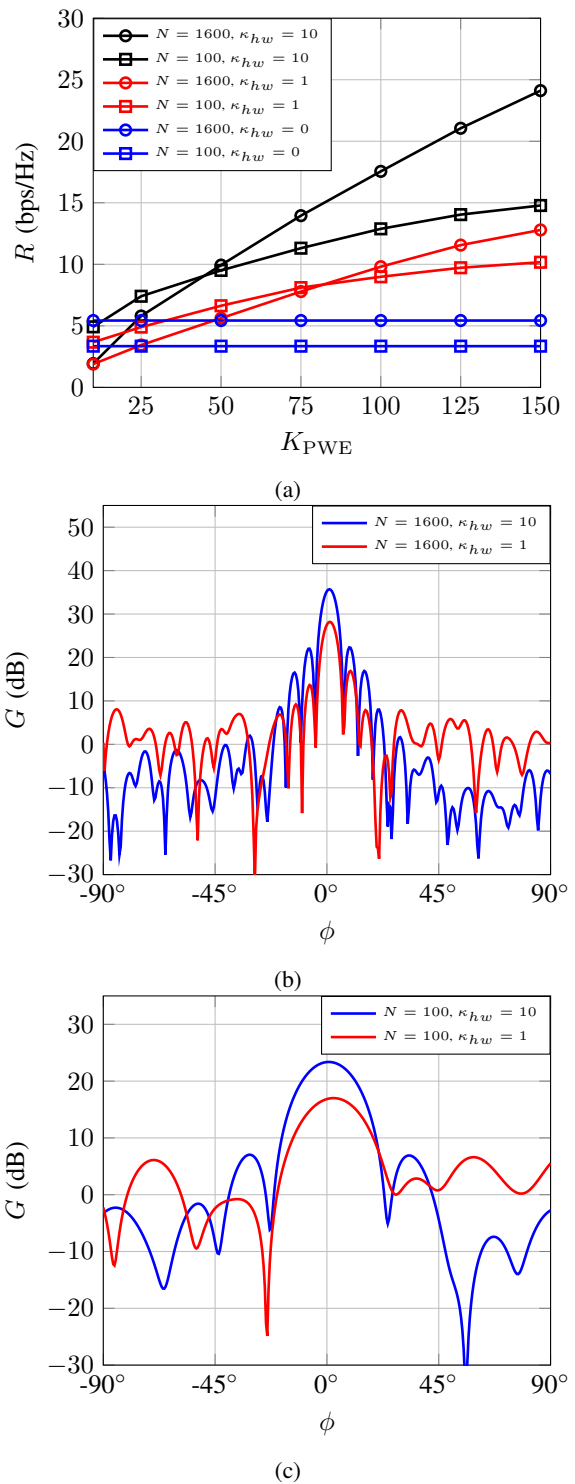
\begin{figure}[t!]
    \centering
    \begin{subfigure}{\linewidth}
        \centering
        \begin{tikzpicture}
            \begin{axis}[
                width=0.85\linewidth,
                xlabel={$K_{\mathrm{PWE}}$},
                ylabel={$R$ (bps/Hz)},
                ymin=0, ymax=30,
                xmin=10, xmax=150,
                xtick={25,50,...,150},
                ytick={0,5,10,15,20,25,30},
                legend cell align={left},
                legend style={font=\tiny, at={(0,1)},anchor=north west},
                grid=major,
                cycle list name=color list
            ]
                \addplot[black, mark=o, mark size=2, line width=1pt, style=solid]
                table {Fig4/10k1600.dat};
                \addlegendentry{$N=1600$, $\kappa_{hw}=10$}
                \addplot[black, mark=square, mark size=2, line width=1pt, style=solid]
                table {Fig4/10k100.dat};
                \addlegendentry{$N=100$, $\kappa_{hw}=10$}
                \addplot[red, mark=o, mark size=2, line width=1pt, style=solid]
                table {Fig4/1k1600.dat};
                \addlegendentry{$N=1600$, $\kappa_{hw}=1$}
                \addplot[red, mark=square, mark size=2, line width=1pt, style=solid]
                table {Fig4/1k100.dat};
                \addlegendentry{$N=100$, $\kappa_{hw}=1$}
                \addplot[blue, mark=o, mark size=2, line width=1pt, style=solid]
                table {Fig4/0k1600.dat};
                \addlegendentry{$N=1600$, $\kappa_{hw}=0$}
                \addplot[blue, mark=square, mark size=2, line width=1pt, style=solid]
                table {Fig4/0k100.dat};
                \addlegendentry{$N=100$, $\kappa_{hw}=0$}
            \end{axis}
        \end{tikzpicture}
        \caption{}
        \label{fig:figure4a}
    \end{subfigure}
    
    \begin{subfigure}
     {\linewidth}
        \centering
        \begin{tikzpicture}
            \begin{axis}[
                width=0.85\linewidth,
                xlabel={$\phi$},
                ylabel={$G$ (dB)},
                xmin=-90, xmax=90,
                ymin=-30, ymax=55,
                xtick={-90, -45, 0, 45, 90},
                xticklabels={-90$^\circ$, -45$^\circ$, 0$^\circ$, 45$^\circ$, 90$^\circ$},
                ytick={-30, -20, -10, 0, 10, 20, 30, 40, 50},
                legend cell align={left},
                legend style={font=\tiny, at={(1,1)},anchor=north east},
                grid=major,
                cycle list name=color list
            ]
                \addplot[blue, no marks, line width=1pt, style=solid]
               table[x index=0, y expr=10*log10(\thisrowno{1})] {Passaloi/10_1600.dat};  
                \addlegendentry{$N=1600$, $\kappa_{hw}=10$}
                \addplot[red, no marks, line width=1pt, style=solid]
                table[x index=0, y expr=10*log10(\thisrowno{1})] {Passaloi/1_1600.dat};
                \addlegendentry{$N=1600$, $\kappa_{hw}=1$}
            \end{axis}
        \end{tikzpicture}
        \caption{}
        \label{fig:figure4b}
    \end{subfigure}

    \begin{subfigure}
    {\linewidth}
        \centering
        \begin{tikzpicture}
            \begin{axis}[
                width=0.85\linewidth,
                xlabel={$\phi$},
                ylabel={$G$ (dB)},
                xmin=-90, xmax=90,
                ymin=-30, ymax=35,
                xtick={-90, -45, 0, 45, 90},
                xticklabels={-90$^\circ$, -45$^\circ$, 0$^\circ$, 45$^\circ$, 90$^\circ$},
                ytick={-30, -20,...,30},
                legend cell align={left},
                legend style={font=\tiny, at={(1,1)},anchor=north east},
                grid=major,
                cycle list name=color list
            ]          
                \addplot[blue, no marks, line width=1pt, style=solid]
                table[x index=0, y expr=10*log10(\thisrowno{1})] {Passaloi/10_100.dat};
                \addlegendentry{$N=100$, $\kappa_{hw}=10$}
                \addplot[red, no marks, line width=1pt, style=solid]
                table[x index=0, y expr=10*log10(\thisrowno{1})] {Passaloi/directivity_data.dat};
                \addlegendentry{$N=100$, $\kappa_{hw}=1$}
            \end{axis}
        \end{tikzpicture}
        \caption{}
        \label{fig:figure4c}
    \end{subfigure}
    \caption{Effect of LeRIS hardware imperfection on (a) Achievable spectral efficiency $R$ for various $K_{\mathrm{PWE}}$ values, (b) LeRIS gain $G$ for $N=1600$, and (c) LeRIS gain $G$ for $N=100$.}
    \label{fig:figure4}
\end{figure}

Finally, in Fig. \ref{fig:figure6}, the impact of the positional offset $\Delta_o$ on spectral efficiency $R$ is analyzed for LeRIS sizes of $N=100$ and $N=1600$, highlighting how such offsets affect performance differently depending on the array size. The PWE configures the RIS as if centered at $(0,0,0)$, while the actual RIS center is displaced by $\Delta_o$. As shown in Fig. \ref{fig:figure6a}, the spectral efficiency decreases for both $\Delta_o=1$ cm and $\Delta_o=3$ cm, with the reduction being more severe for $N=1600$ due to its narrower beamwidth, which increases sensitivity to misalignment. This effect becomes more evident when examining the LeRIS gain $G$ in Figs. \ref{fig:figure6b} and \ref{fig:figure6c}, which show the gain patterns for $N=1600$ and $N=100$, respectively, with the arrays configured to steer toward $\phi=45^\circ$ and $\theta=30^\circ$. For both $N=1600$ and $N=100$, the main lobe shifts away from the target direction as $\Delta_o$ increases, but the impact is more pronounced for $N=1600$ because its high-gain area is more concentrated. This makes even a small misalignment translate into a significant loss of gain at the intended direction, which can lead to outages for small $K_{\mathrm{PWE}}$ values. In contrast, $N=100$, with its broader beam, spreads energy across a wider angular region, preserving some gain near the target even when an offset is present, which helps maintain better performance despite positional inaccuracies. As $K_{\mathrm{PWE}}$ increases, however, improved localization accuracy reduces the steering error caused by $\Delta_o$, allowing the UE to remain within the high-gain area. Nevertheless, for larger $\Delta_o$, even high localization accuracy may not fully compensate for the lobe misalignment, emphasizing the importance of precisely knowing the RIS position within the PWE, especially for larger arrays. In practice, though, these offsets are typically small in controlled environments, and an efficient LATC scheme like the one proposed can effectively mitigate such errors, ensuring PWE robust performance.

\begin{figure}[t!]
    \centering
    \begin{subfigure}{0.85\linewidth}
        \centering
    \begin{tikzpicture}
        \begin{axis}[
            width=0.99\linewidth,
            xlabel = {$K_{\mathrm{PWE}}$},
            ylabel = {$R$ (bps/Hz)},
            ymin = 0,
            ymax = 30,
            xmin = 10,
            xmax = 250,
            xtick = {50,100,...,250},
            ytick = {0,5,...,30},
            grid = major,
		legend entries ={{$N=100$, $\Delta_o= 1$ cm},{$N=1600$, $\Delta_o= 1$ cm},{$N=100$, $\Delta_o= 3$ cm},{$N=1600$, $\Delta_o= 3$ cm} },
            legend cell align = {left},
            legend style={font=\tiny},
            legend style={at={(0,1)},anchor=north west},
            ]
            \addplot[
            blue,
            mark = o,
            mark repeat = 1,
            mark size = 2,
            line width = 1pt,
            style = solid,
            ]
            table {Fig6/7_100.dat};
            \addplot[
            blue,
            mark = square,
            mark repeat = 1,
            mark size = 2,
            line width = 1pt,
            style = solid,
            ]
            table {Fig6/7_1600.dat};

            \addplot[
            red,
            mark = o,
            mark repeat = 1,
            mark size = 2,
            line width = 1pt,
            style = solid,
            ]
            table {Fig6/2_100.dat};
            \addplot[
            red,
            mark = square,
            mark repeat = 1,
            mark size = 2,
            line width = 1pt,
            style = solid,
            ]
            table {Fig6/2_1600.dat};
            \end{axis}
    \end{tikzpicture}
    \caption{
    }
    \label{fig:figure6a}
    \end{subfigure}
 \begin{subfigure}{\linewidth}
        \centering
        \includegraphics[width=\textwidth]{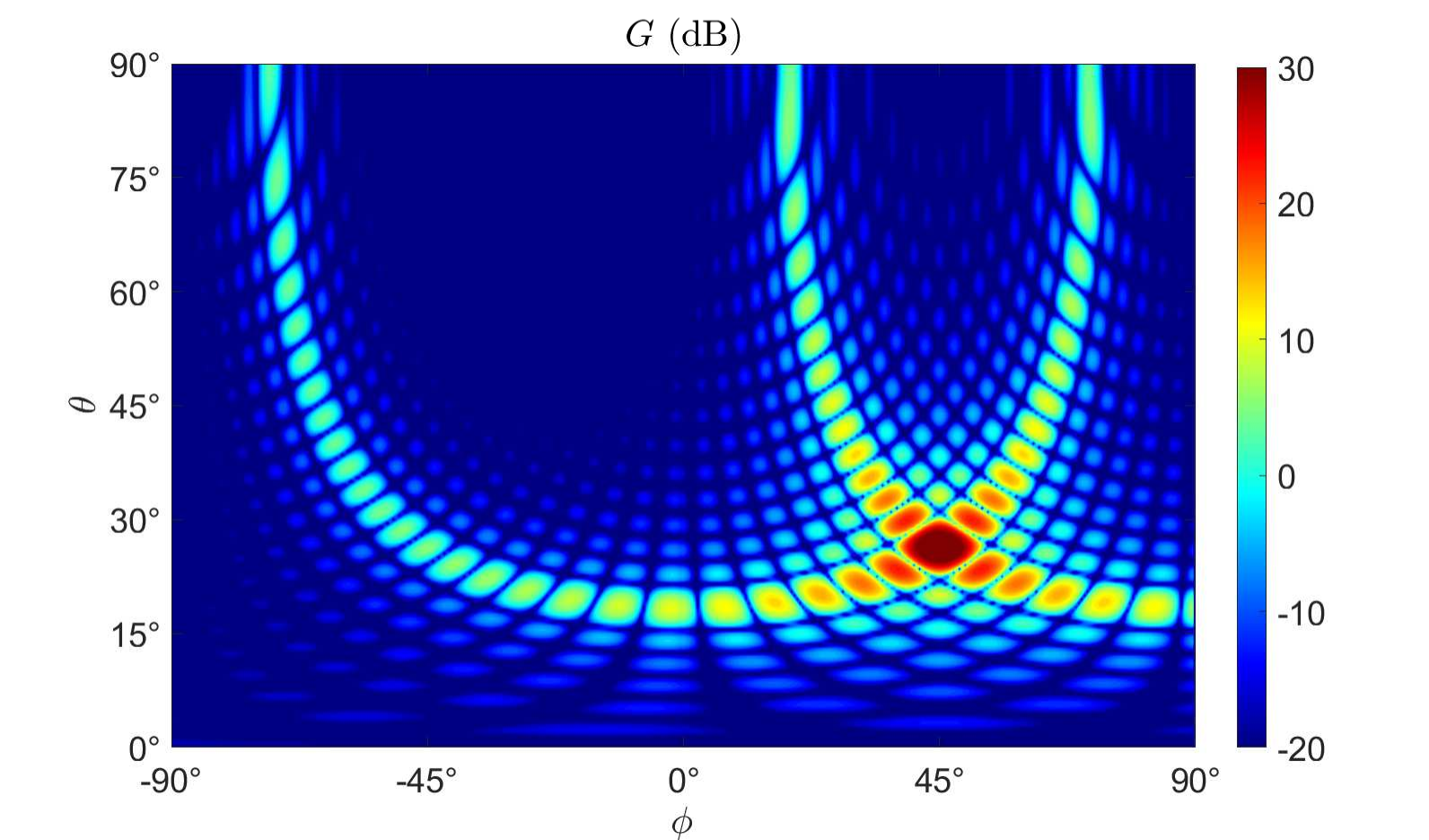}
        \caption{}
        \label{fig:figure6b}
    \end{subfigure}
     \begin{subfigure}{\linewidth}
        \centering
        \includegraphics[width=\textwidth]{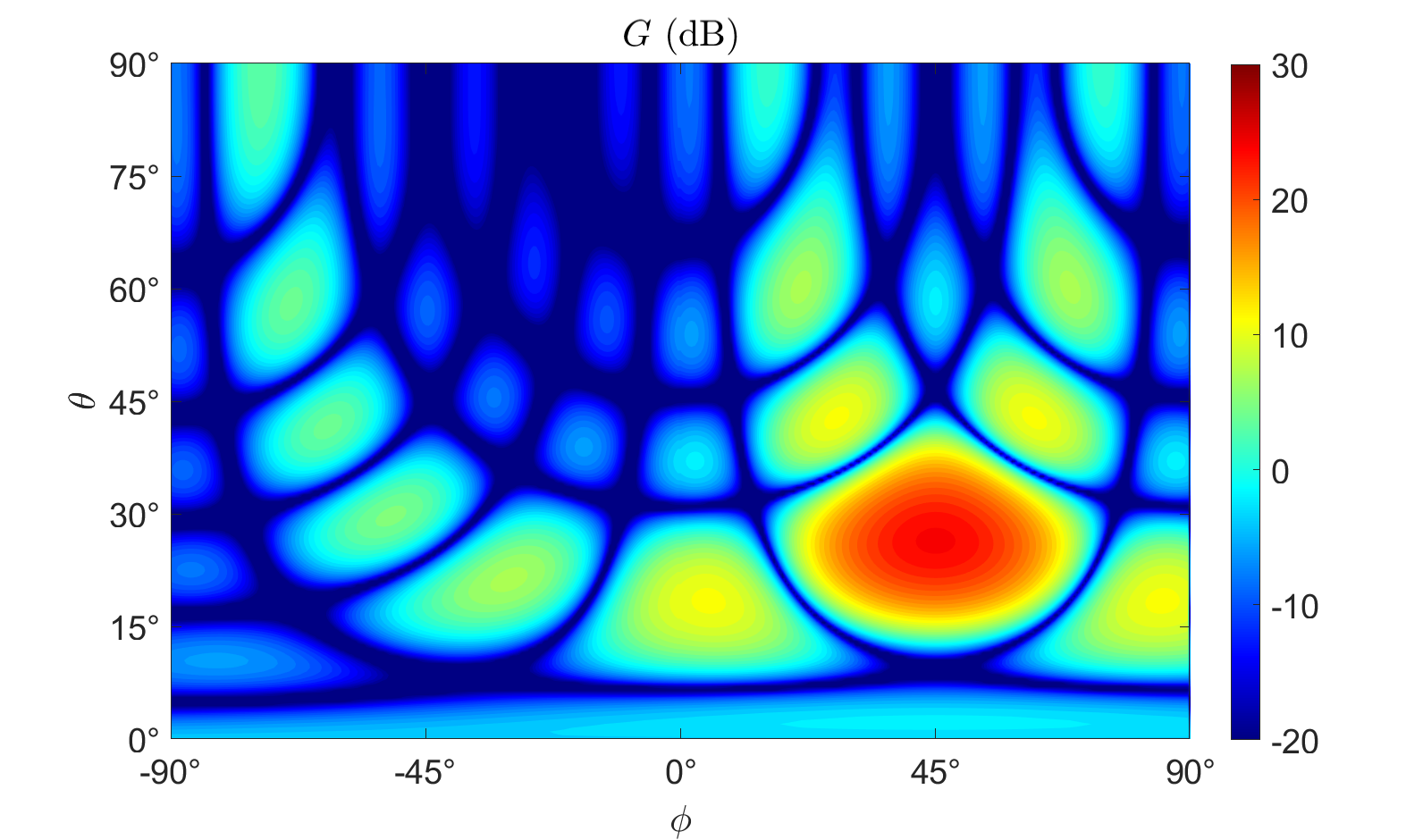}
        \caption{}
        \label{fig:figure6c}
    \end{subfigure}
    \caption{Effect of position offset $\Delta_o$ on (a) Achievable spectral efficiency $R$ for various $K_{\mathrm{PWE}}$ values, (b) LeRIS Gain $G$ for $N=1600$, and (c) LeRIS Gain $G$ for $N=100$.}
    \label{fig:figure6}
\end{figure}
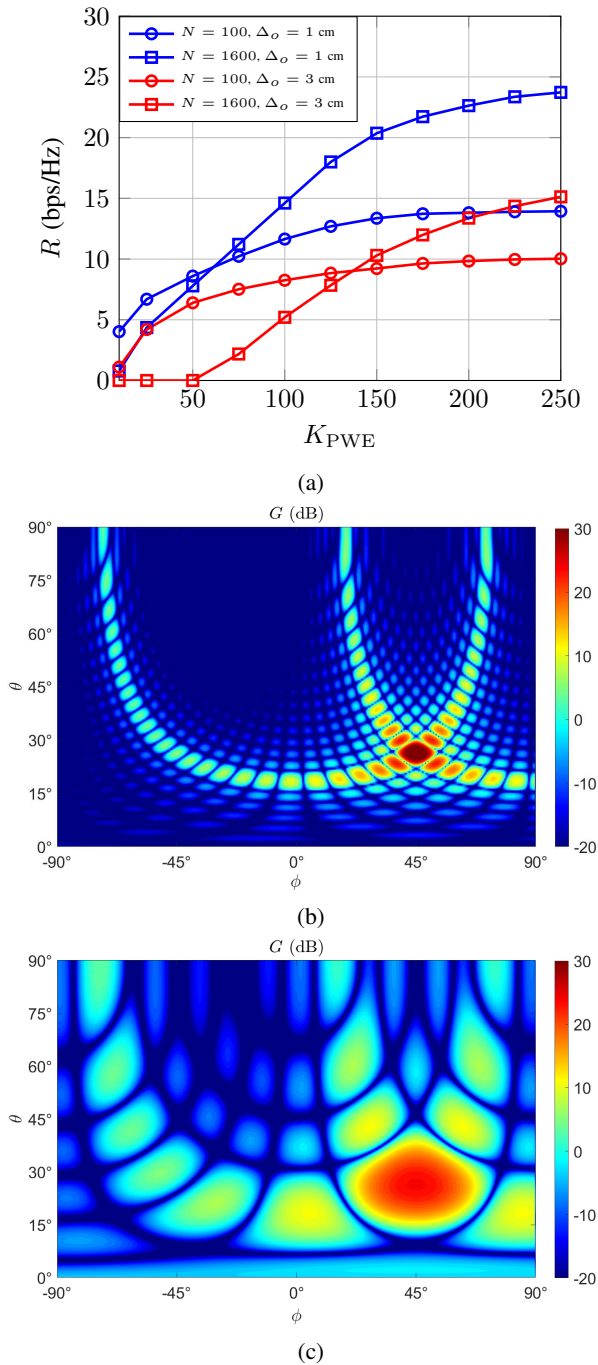

\section{Conclusions}\label{S6}
In this work, we proposed an LED deployment strategy for PWEs that combines ceiling-mounted LEDs with LeRISs to achieve RSS-based localization for receivers with random orientations. By leveraging the spatial diversity of the LeRIS architecture, the proposed framework creates resilient signal paths that reduce localization errors caused by varying UE orientations. To support this, we derived closed-form expressions for optical RSS-based localization that account for random receiver orientations and established spatial constraints on LED placement to ensure uniquely solvable and accurate localization. Through extensive simulations, we demonstrated that the proposed LATC scheme not only delivers precise localization but also supports high spectral efficiency and robust beam control, even in the presence of hardware imperfections and positional offsets. Notably, for larger LeRIS sizes, where narrow directional beams are more sensitive to misalignments, the LATC scheme compensates effectively, ensuring performance under realistic conditions. In summary, the proposed localization-driven framework provides a scalable and adaptable solution for advancing PWE applications, demonstrating its ability to overcome practical challenges such as hardware imperfections and positional offsets. By combining precise localization with resilient beam control, this work paves the way for deploying robust and efficient RIS-assisted systems capable of meeting the high precision and responsiveness demands of PWEs and 6G networks.

\section*{Appendix A - Proof of Lemma \ref{l1}}
According to the optical RSS-based method, the UE position can be derived by solving a system of equations based on the RSS measurements for each LED that has a LoS link with the UE. Thus, considering 4 LEDs belonging to the set $\mathcal{I}$, and without loss of generality are placed upon $y=y_i$ plane, we obtain four equations from \eqref{coor}, each relating the RSS value to the UE’s spatial coordinates and the corresponding angle of incidence $\psi_{a,i}$. However, since this system involves seven unknown variables, i.e., the spatial coordinates of the UE and the four incidence angles $\psi_{a,i}$, we need three additional equations to determine all these unknowns. These additional equations can be derived from geometric equations between the LED positions and the UE coordinates, involving the incidence angles $\psi_{a,i}$ and the spatial arrangement of the LEDs. However, certain constraints naturally emerge through the RSS equations, restricting the possible geometric arrangements of the LEDs. Therefore, it is essential to identify these deployment constraints to ensure the system remains solvable.

To simplify our analysis and without loss of generality, we assume $y_i = 0$, which simplifies \eqref{coor} to
    \begin{align}\label{eq2}
    &\frac{(\hat{x}_u-x_i)^2 + \hat{y}_u^2 + (\hat{z}_u - z_i)^2}{\hat{y}_u^{\frac{2m}{m_l+2}}}    \nonumber \\ 
    &  = \left (\frac{P_{o,i}}{P_{r,i}} \frac{(m_l+1)A_{\mathrm{PD}}}{2\pi}T_{f}g_{c}(\psi_{a,i})\cos(\psi_{a,i})\right)^{2/(m_l+2)}.
\end{align}
Additionally, we assume that the PD can receive the signal from 4 LEDs and that the angle of incidence $\psi_{a,i}$ from each LED is within the maximum allowable field of view, ensuring $g_c(\psi_{a,i}) \neq 0, \forall i \in \mathcal{I}$.

To express the UE coordinates as a function of the system parameters, we start by dividing \eqref{eq2} for $i=1$ and $i=2$, yielding
\begin{equation}\label{sub1}
    \frac{(\hat{x}_u-x_1)^2 + \hat{y}_u^2 + (\hat{z}_u - z_1)^2}{(\hat{x}_u-x_2)^2 + \hat{y}_u^2 + (\hat{z}_u - z_2)^2 } = \left( \frac{P_{r,2} \cos(\psi_{a,1})}{P_{r,1} \cos(\psi_{a,2})} \right)^{2/(m_l+2)}.
\end{equation}
It should be noted that, to perform the division of these equations, the UE must not lie on the same plane as the LEDs, i.e., $\hat{y}_u \neq 0$, as evident from \eqref{eq2}. Moreover, to simplify the notation, we set
\begin{equation}\label{a1}
    \alpha_{1}= \left(\frac{P_{r,2} \cos(\psi_{a,1})}{P_{r,1} \cos(\psi_{a,2})} \right)^{2/(m_l+2)},
\end{equation}
and after some algebraic manipulations, we can express $\hat{y}_u$ as \eqref{l} which is shown at the top of the page. 
It should be noted, that since the direction of the LED transmission coincides with the positive $y$-axis, $\hat{y}_u$ must be a positive value, and thus, the negative root of \eqref{sub1} is disregarded.

Next, to express $\hat{x}_u$ as a function of the system parameters, similarly as with $\hat{y}_u$, we can divide by \eqref{eq2} for $i=1$ and $i=3$, and by substituting \eqref{l} and setting
\begin{equation}\label{a2}
    \alpha_{2} = \left(\frac{P_{r,3} \cos(\psi_{a,1})}{P_{r,1} \cos(\psi_{a,3})} \right)^{2/(m_l+2)},
\end{equation}
we obtain \eqref{sub2.1} which is given at the top of this page. Thus, by manipulating \eqref{sub2.1}, $\hat{x}_u$ can be obtained as
\begin{figure*}
    \begin{equation}\label{l}
    \hat{y}_u = \sqrt{\frac{1}{\alpha_{1}-1} \left( (\hat{x}_u - x_1)^2 + (\hat{z}_u - z_1)^2  -\alpha_{1}\left((\hat{x}_u - x_2)^2 + (\hat{z}_u - z_2)^2\right) \right)}
    \end{equation} 
\hrule    
\end{figure*}
\begin{figure*}
\begin{equation}\label{sub2.1}
\begin{aligned}
     & (\hat{x}_u - x_1)^2 + (\hat{z}_u - z_1)^2 + \frac{1}{\alpha_{1} -1} \left[ (\hat{x}_u - x_1)^2 
     + (\hat{z}_u - z_1)^2 - 
      \alpha_{1} \left((\hat{x}_u - x_2)^2 + (\hat{z}_u - z_2)^2\right) \right] \\
     & \qquad = \alpha_{2}\left[(\hat{x}_u - x_3)^2 + (\hat{z}_u - z_3)^2\right] 
      + \frac{\alpha_{2}}{\alpha_{1} -1} \left[ (\hat{x}_u - x_1)^2 + (\hat{z}_u - z_1)^2 - 
    \alpha_{1}  \left(  (\hat{x}_u - x_2)^2 + (\hat{z}_u - z_2)^2 \right)\right]
\end{aligned}
\end{equation}
\hrule
\end{figure*}
\begin{equation}\label{x}
    \hat{x}_u = \frac{\xi_1\hat{z}_u + \xi_3}{\xi_2},
\end{equation}
where $\xi_1$, $\xi_2$, and $\xi_3$ are shown at the top of this page in \eqref{xi1}, \eqref{xi2}, and \eqref{xi3}, respectively.
\begin{figure*}
\begin{align}
    &\xi_1 = 2\left [\alpha_{1}(z_1 - z_2) + \alpha_{2}(z_3 - z_1) + \alpha_{1}\alpha_{2}(z_2 - z_3) \right], \label{xi1} \\
    &\xi_2 = 2\left [\alpha_{1}(x_2 - x_1) + \alpha_{2}(x_1 - x_3) + \alpha_{1}\alpha_{2}(x_3 - x_2) \right],\label{xi2} \\
    &\xi_3 = \alpha_{1}\left(x_2^2-x_1^2 + z_2^2 - z_1^2\right) + \alpha_{2}\left(x_1^2 - x_3^2 + z_1^2 - z_3^2\right) 
    + \alpha_{1}\alpha_{2}\left( x_3^2 - x_2^2 + z_3^2 - z_2^2 \right) \label{xi3}
\end{align}
\hrule
\end{figure*}

Finally, to express $\hat{z}_u$ as a function of the system parameters, we divide \eqref{eq2} for $i=1$ and $i=4$,  substitute \eqref{l} and \eqref{x}, and then set $\alpha_3$ as 
\begin{equation}\label{a3}
   \alpha_3 = \left(\frac{P_{r,4} \cos(\psi_{a,1})}{P_{r,1} \cos(\psi_{a,4})} \right)^{2/(m_l+2)},
\end{equation}
which leads to \eqref{init} which is given at the top of the next page, and by manipulating \eqref{init}, we obtain \eqref{z} given at the top of the next page.
\begin{figure*}
    \begin{align}\label{init}
            &\left(  \frac{\xi_1 \hat{z}_u + \xi_3}{\xi_2} - x_1 \right)^2 + \left(\hat{z}_u - z_1\right)^2 + \frac{1}{\alpha_{1}-1}\left[   \left(  \frac{\xi_1 \hat{z}_u + \xi_3}{\xi_2} - x_1 \right)^2 + \left(\hat{z}_u - z_1\right)^2 -\alpha_{1}\left( \left( \frac{\xi_1 \hat{z}_u + \xi_3}{\xi_2} - x_2 \right)^2 + \left(\hat{z}_u - z_2\right)^2\right) \right]  \nonumber \\ 
            &\qquad = \alpha_{3} \left[\left(  \frac{\xi_1 \hat{z}_u + \xi_3}{\xi_2} - x_4 \right)^2 + \left(\hat{z}_u - z_4\right)^2\right]
            +\frac{\alpha_{3}}{\alpha_{1}-1}\left[   \left(  \frac{\xi_1 \hat{z}_u + \xi_3}{\xi_2} - x_1 \right)^2 + \left(\hat{z}_u - z_1\right)^2 \nonumber \right.\\
            & \qquad \qquad \left. -\alpha_{1}\left( \left( \frac{\xi_1 \hat{z}_u + \xi_3}{\xi_2} - x_2 \right)^2 + \left(\hat{z}_u - z_2\right)^2\right) \right]  
    \end{align} 
\hrule
        \begin{align}\label{z}
            \hat{z}_u=&  \frac{\left( \alpha_{1} - \alpha_{3} \right) \left[ x_1\left(\xi_2x_1 -2\xi_3\right) + \xi_2 z_1^2 \right] + \alpha_{1} \left( \alpha_{3} - 1 \right) \left[ x_2\left(\xi_2x_2 -2\xi_3\right) + \xi_2 z_2^2 \right] - \alpha_{3} \left( \alpha_{1} - 1 \right) \left[ x_4\left(\xi_2x_4 -2\xi_3\right) + \xi_2 z_4^2 \right]}{2\left[\left( \alpha_{1} - \alpha_{3} \right)\left( \xi_1x_1 + \xi_2z_1 \right) + \alpha_{1}\left( \alpha_{3} - 1 \right)\left( \xi_1x_2 + \xi_2z_2 \right) - \alpha_{3} \left( \alpha_{1} - 1 \right) \left( \xi_1x_4 + \xi_2z_4 \right)\right]}
        \end{align}
\hrule
\end{figure*}

Since we have now expressed the UE coordinates as functions of the system parameters, some constraints naturally arise regarding the LED positions. In more detail, based on \eqref{l} it, must hold that $\alpha_{1} \neq 1$, which translates to
\begin{equation}\label{con1}
    \frac{\cos{(\psi_{a,1})}}{P_{r,1}} \neq \frac{\cos{(\psi_{a,2})} } {P_{r,2}}. 
\end{equation}
Furthermore, from \eqref{x}, we can observe that $\xi_2$ must not be $0$, which only holds when at least 3 out of 4 LEDs do not have the same $x$ coordinate, since $a_j > 0, \, j \in \{1,2,3 \}$. To address these constraints regarding the identical $x$ coordinate, \eqref{sub2.1} should be solved for $\hat{z}_u$ rather than $\hat{x}_u$, thereby shifting the constraint to the $z$ coordinate of the LEDs instead.

Finally, to ensure that the system can be solved, the denominator of \eqref{z} can also not be equal to 0. More specifically, by taking into account that $a_1 \neq 1$, $\xi_2 \neq 0$, then to assure that the denominator in \eqref{z} is not zero, it must hold that $a_1 \neq a_3$ and $a_3 \neq 1$. Therefore, after some algebraic manipulations, the constraints mentioned above can be rewritten as 
\begin{equation}\label{con2}
    \frac{\cos{(\psi_{a,2})}}{P_{r,2}} \neq \frac{\cos{(\psi_{a,4})} } {P_{r,4}},
\end{equation}
and
\begin{equation}\label{con3}
    \frac{\cos{(\psi_{a,1})}}{P_{r,1}} \neq \frac{\cos{(\psi_{a,3})} } {P_{r,3}}.
\end{equation}
From constraints \eqref{con1}, \eqref{con2}, and \eqref{con3}, it can be derived that the system can be solved if and only if the distances from the LEDs, which are located on a plane parallel to the UE plane, to the UE are not equal. To prove that, these constraints must be expressed as functions of the distances between the LEDs and the UE for the case where the UE is parallel to the LEDs. Due to the planes being parallel, the angles $\psi_{d,i}$ are equal to $\psi_{a,i}$, thus, we can calculate the ratio of the received powers, appeared on \eqref{con1}, \eqref{con2}, and \eqref{con3}, as a function of the distances from each LED, by substituting $\psi_{a,i}$ with \eqref{AoD}, and then by solving \eqref{d_l} for $P_{r,i}$. 
Thus, by dividing \eqref{con1}, \eqref{con2}, and \eqref{con3} pairwise and performing algebraic manipulations, the constraints can be rewritten as
\begin{equation}\label{dist1}
    d_{l,1} \neq d_{l,2},
\end{equation}
and
\begin{equation}\label{dist2}
    d_{l,1} \neq d_{l,3},
\end{equation}
and
\begin{equation}\label{dist3}
    d_{l,2}\neq d_{l,4}.
\end{equation}
Thus, from \eqref{dist1}, \eqref{dist2}, and \eqref{dist3}, we can determine that the LEDs that are parallel to the UE must not have equal distances from the UE, which concludes the proof.

\bibliographystyle{IEEEtran}
\bibliography{Bibliography}
\end{document}